\documentclass[10pt,twocolumn,twoside]{IEEEtran}
\IEEEoverridecommandlockouts

\usepackage{tikz}
\usetikzlibrary{positioning, arrows.meta, shapes.geometric}
\usepackage{multirow}

\usepackage{float}
\usepackage{booktabs}
\usepackage{threeparttable}
\usepackage{array}
\usepackage{url}
\usepackage{cite}
\usepackage{pifont}

\usepackage{amsmath,amssymb,amsfonts,amsthm,mathtools}
\usepackage{graphicx, subcaption}
\usepackage{textcomp}
\usepackage{balance}
\usepackage{xcolor}
\usepackage[utf8]{inputenc}
\usepackage{algorithm, algpseudocode}
\newtheorem{proposition}{Proposition}
\theoremstyle{plain} 
\newtheorem{theorem}{Theorem}[section]   
\newtheorem{lemma}[theorem]{Lemma}

\theoremstyle{definition} 
\newtheorem{definition}[theorem]{Definition}

\DeclareMathOperator*{\argmin}{arg\,min}

\def\BibTeX{{\rm B\kern-.05em{\sc i\kern-.025em b}\kern-.08em
    T\kern-.1667em\lower.7ex\hbox{E}\kern-.125emX}}
\usepackage{lipsum}

\title{\LARGE \bf 
Strategic Spatial Load Shifting and Market Efficiency
}

\author{Aron Brenner$^*$, Deepjyoti Deka$^\dagger$, Line Roald$^\ddagger$, Saurabh Amin$^{*}$
\vspace{-0.5cm}
\thanks{ 
$^*$Aron Brenner and Saurabh Amin are with the Laboratory for Information and Decision Systems, MIT, Cambridge, MA, USA. Emails: {\tt \{abrenner, amins\}@mit.edu}}
\thanks{ 
$^\dagger$Deepjyoti Deka is with the MIT Energy Initiative, MIT, Cambridge, MA, USA. Email: {\tt deepj87@mit.edu}}
\thanks{$^\ddagger$Line Roald is with the Department of Electrical and Computer Engineering, University of Wisconsin-Madison, Madison, WI, USA. Email: {\tt roald@wisc.edu}}
}

\begin{document}
\begingroup
\allowdisplaybreaks

\maketitle

\begin{abstract}
Large, spatially flexible electricity consumers such as data centers can reallocate demand across locations, influencing dispatch and prices in wholesale electricity markets. While flexible load is often assumed to improve system efficiency, this intuition typically relies on price-taking behavior. We study price-anticipatory spatial load shifting by modeling a large flexible consumer as a Stackelberg leader interacting with DC optimal power flow (DC-OPF) based market clearing. We show that decentralized, cost-minimizing load shifting need not align with system operating cost minimization, and that misalignment arises at boundaries between DC-OPF operating regimes, where small changes in load can induce discrete changes in marginal generators or congestion patterns. We evaluate strategic load shifting on the 73-bus RTS-GMLC test system, where findings indicate reductions in system operating cost in most hours, but misalignment in a subset of cases that are driven by redispatch at merit-order discontinuities. We find that these outcomes are primarily redistributive relative to a price-taking benchmark, reducing generator profits while lowering electricity procurement costs for both flexible and inflexible consumers, even in cases where total system operating costs increase.
\end{abstract}


\section{Introduction}
After two decades of relatively flat demand, electricity consumption in the United States is rising again, driven in large part by the rapid expansion of energy-intensive data centers \cite{eia2025}. Data centers accounted for roughly 4.4\% of total U.S. electricity consumption in 2023, a share projected to increase to up to 12\% by 2028 \cite{shehabi20242024}. Beyond their scale, a defining feature of large data centers is operational flexibility: computational workloads can be shifted across locations and time, allowing operators to reallocate electricity demand in response to prices, emissions, or other system signals. Such flexibility is already being actively exploited by hyperscale operators to reduce procurement costs and carbon footprints \cite{google_env_report,Meta_sust_report,ZhengEtal2020_joule,radovanovic2022carbon, SouzaEtal2023}.

A common intuition is that flexible demand improves system efficiency by shifting consumption toward locations or hours with abundant low-cost, low-carbon supply---effectively acting as ``virtual transmission'' that relieves congestion and increases renewable utilization \cite{ZhengEtal2020_joule,zhang2020flexibility}. Accordingly, many studies argue that demand flexibility can play a central role in the grid interconnection of large flexible loads, reducing the need for additional investment in generation and transmission capacity \cite{farmer2025doe,norris2025rethinking}. However, recent work highlights that the system-level consequences of decentralized flexibility are highly context dependent: long-run planning studies find possible cost reductions but ambiguous emissions impacts driven by investment responses \cite{senga2025}, while short-run studies show that emissions-guided shifting based on imperfect signals can increase total system emissions \cite{GorkaEtal2024}.

While previous literature raises potential inefficiencies arising from imperfect emissions signals or long-run investment effects, our focus is different: we isolate the role of \textit{strategic behavior}, treating large flexible consumers as \textit{price-anticipatory} participants in market clearing whose actions can alter marginal generators and congestion patterns, rather than as price-takers. In particular, we ask when such decentralized, strategic behavior reduces total system operating cost, and when it does not.

This work aims primarily to address the following question:
\begin{center}
    \textit{Through what mechanisms does decentralized, strategic spatial load shifting align with or conflict with the system-level cost objective?}
\end{center}
Our contributions are as follows.

\textbf{1. Analytical characterization of misalignment.}
We derive a necessary condition under which price-anticipatory spatial load shifting leads to misalignment between private procurement incentives and system operating cost minimization. We show that misalignment cannot arise within a single DC-OPF operating regime (Proposition~\ref{prop:sufficient_condition_for_alignment}) and occurs only when an optimal load shift positions the system at the boundary between active-set regions where marginal generators or binding transmission constraints change (Proposition~\ref{prop:necessary_condition_for_misalignment}).

\textbf{2. Mechanism-based illustrative example.}
We develop a three-zone network example that makes operating regime transitions and price discontinuities explicit, illustrating how strategic load shifting can reduce consumer procurement costs while increasing total system operating cost through discrete redispatch effects.

\textbf{3. Computational evaluation on benchmark test system.}
We evaluate price-anticipatory spatial load shifting on the 73-bus RTS-GMLC test system, showing that decentralized load shifting is aligned with system objectives in most hours, but increases system operating costs in a subset of cases driven by redispatch at merit-order discontinuities.

\section{Problem Formulation}
\label{sec:preliminaries}

We model the interaction between a single, spatially flexible consumer and the electricity market as a Stackelberg game in which the flexible consumer acts as the leader and the market operator as the follower. The consumer is assumed to have perfect information about the market clearing process and chooses load shifts anticipating their impact on dispatch and prices, subject to physical flexibility constraints. Following the consumer’s decision, the operator clears the market by solving a deterministic DC optimal power flow (DC-OPF) problem, yielding dispatch outcomes, system operating costs, and locational marginal prices (LMPs).

The flexible consumer represents a large load (such as a hyperscaler) whose spatial reallocation can affect congestion patterns and marginal generators, and is therefore modeled as price-anticipatory rather than price-taking. By contrast, the operator is assumed to be non-strategic and applies fixed market clearing rules without responding to consumer incentives on the short time scales considered here. Our model extends prior work on two-zone systems \cite{brenner2026strategicdatacenterload} to general network topologies with multiple congestion patterns. This bilevel formulation also shares the price-anticipatory structure of \cite{kazempour2015}; however, our modeling of spatial flexibility makes the two fundamentally different. There, a consumer exercises market power by \textit{bidding} below its marginal utility and withholding consumption to depress prices, whereas in our formulation, total demand is fixed by a load-shifting balance constraint, so market power is exercised purely through \textit{spatial reallocation}.

\subsection{Lower Level: DC-OPF}
\label{subsec:dcopf}

The lower-level problem models wholesale market clearing via a DC-OPF problem. The transmission network consists of buses $\mathcal{N}=\{1,\dots,n\}$ and transmission lines $\mathcal{L}=\{1,\dots,m\}$, with oriented node–edge incidence matrix $A\in\mathbb{R}^{m\times n}$ and diagonal susceptance matrix $B=\mathrm{diag}(b_1,\dots,b_m)$. Line flow limits are given by $\bar F\in\mathbb{R}^m_+$, and the reference (slack) bus is denoted by index $s\in\mathcal{N}$.

Generation is represented by a fleet of $k$ generators with marginal costs $c\in\mathbb{R}^k$ and capacity limits $\bar p\in\mathbb{R}^k_+$. The matrix $G\in\mathbb{R}^{n\times k}$ maps generators to buses, where $G_{ij}=1$ if generator $j$ injects at bus $i$. System load is given by a nominal demand vector $d\in\mathbb{R}^n_+$, decomposed into inflexible base load $d_{\mathrm{base}}$ and nominal flexible load $d_{\mathrm{flex}}$. For a load shift $\delta\in\mathbb{R}^n$, market clearing is obtained by solving a DC-OPF with realized load $d+\delta$:
\begin{subequations}\label{eq:dcopf}
\begin{align}
    V(\delta)
    \coloneqq
    \min_{p,\theta,f}
    \quad & c^\top p \label{eq:dcopf_new}\\
    \mathrm{s.t.}\quad & f = B A \theta & (\eta) \label{eq:flow_angles}\\
    & A^\top f + d + \delta = G p & (\lambda) \label{eq:energy_balance} \\
    & -\bar F \leq f \leq \bar F & (\mu^+, \mu^-) \label{eq:flow_limits} \\
    & 0 \leq p \leq \bar p & (\pi^+, \pi^-) \label{eq:dispatch_limits} \\
    & e_s^\top \theta = 0 & (\nu) \label{eq:dcopf_refbus}
\end{align}
\end{subequations}
where $p \in \mathbb{R}^k$ denotes offer segment dispatch values, $\theta \in \mathbb{R}^n$ the bus voltage angles, and $f \in \mathbb{R}^m$ the transmission line flows. The optimal value $V(\delta)$ represents total system operating cost. Constraint \eqref{eq:flow_angles} enforces DC power flow physics, constraint~\eqref{eq:energy_balance} enforces energy balance at the bus level, constraints \eqref{eq:flow_limits} and \eqref{eq:dispatch_limits} impose transmission and generation capacity limits, and constraint \eqref{eq:dcopf_refbus} fixes the reference bus angle. Since \eqref{eq:dcopf_new} is a linear program, the value function $V(\cdot)$ is convex and piecewise affine in $\delta$, and the corresponding LMP mapping is piecewise constant.

Dual variables corresponding to each constraint are shown in parentheses. In particular, $\lambda \in \mathbb{R}^n$ denotes the vector of LMPs, equal to the shadow prices of the nodal power balance constraints \eqref{eq:energy_balance} and representing the marginal cost of serving an additional unit of load at each bus\footnote{
At operating points where the DC-OPF is degenerate, the LMP vector is not unique. We interpret $\lambda(\delta)$ as a dual optimal solution minimizing the flexible consumer's procurement cost among all admissible LMP vectors, i.e.,
\begin{align*}
    \lambda(\delta) \in \argmin_{\lambda \in \Lambda(\delta)} \lambda^\top\big(d_{\mathrm{flex}} + \delta\big),
\end{align*}
where $\Lambda(\delta)$ denotes the set of dual optimal solutions at load $d+\delta$. This selection is without loss of generality and can be obtained as the limit of arbitrarily small perturbations that restore dual uniqueness.
}. Remaining dual variables capture shadow prices of DC flow physics, transmission, generation limits, and the reference bus.

\subsection{Upper Level: Spatially Flexible Load}
\label{subsec:flex}
\textbf{Load shifting constraints.}
Recall that we model a spatially flexible consumer through an additive load shift $\delta$ applied to a nominal flexible demand profile $d_{\mathrm{flex}}$. For a given flexibility level $\alpha \in [0,1]$, admissible load shifts are restricted to a polyhedral feasible set
\begin{align}
    \mathcal{F}(\alpha)
    \coloneqq
    \Big\{
    \delta \in \mathbb{R}^n :
    T(\alpha)\delta \leq q(\alpha),\;
    1^\top \delta = 0
    \Big\},
    \label{eq:polyhedral_shift_set}
\end{align}
where the balance constraint $1^\top \delta = 0$ ensures that load shifting preserves the flexible consumer's total energy demand. We assume throughout that $q(\alpha) \geq 0$, so that the no-shift action $\delta = 0$ is feasible. While our theoretical results apply to general polyhedral sets of the form \eqref{eq:polyhedral_shift_set}, we highlight as a concrete and interpretable special case the following ``box-with-balance'' model:
\begin{align}
    \mathcal{F}_\mathrm{box}(\alpha)
    =
    \Big\{
    \delta \in \mathbb{R}^n :
    |\delta| \leq \alpha d_{\mathrm{flex}},\;
    1^\top \delta = 0
    \Big\},
    \label{eq:box_with_balance}
\end{align}
which assumes that the load at a bus can increase or decrease by a fraction $\alpha$ of its nominal flexible load.

\textbf{Bilevel load shifting problem.} Let $\lambda(\delta)$ denote the LMP vector corresponding to the market clearing solution as a function of the load shift $\delta$. The flexible consumer procurement cost is then defined as
\begin{align}
    \Pi(\delta) \coloneqq \lambda(\delta)^\top\big(d_\mathrm{flex} + \delta\big),
\end{align}
and the corresponding flexible consumer cost minimizing problem is
\begin{subequations}\label{eq:load_shifting}
\begin{align}
    \min \quad & \lambda(\delta)^\top\big(d_\mathrm{flex} + \delta\big) \label{eq:load_shifting_obj} \\
    \mathrm{s.t.} \quad & \lambda(\delta) = \mathrm{Dual}(\text{DC-OPF}(\delta)) \label{eq:LMP} \\
    & \delta  \in \mathcal{F}(\alpha) \label{eq:feasible_shift}.
\end{align}
\end{subequations}
Here, constraint \eqref{eq:LMP} relates the LMPs $\lambda(\delta)$ to the shadow prices of the solution of the DC-OPF problem instantiated with a load shift of $\delta$ while constraint \eqref{eq:feasible_shift} defines the feasible region for load shifts $\delta$.

Importantly, Problem~\eqref{eq:load_shifting} is a \emph{bilevel program}, reflecting the leader–follower structure between the flexible consumer and the market operator: the consumer chooses a load shift~$\delta$ anticipating the endogenous price response induced by DC-OPF market clearing. To solve Problem~\eqref{eq:load_shifting} in practice, we replace the lower-level DC-OPF with its optimality conditions, yielding a single-level bilinear reformulation that preserves price endogeneity and can be solved using spatial branch-and-bound methods (see Appendix~\ref{app:reformulation} for details).

Crucially, the bilevel structure of \eqref{eq:load_shifting} exposes a fundamental tension between private and system objectives. Because the LMP mapping $\lambda(\delta)$ is piecewise constant, the flexible consumer's objective in \eqref{eq:load_shifting_obj} is inherently discontinuous and differs qualitatively from the system operator's objective \eqref{eq:dcopf_new}. As a result, a load shift that minimizes the flexible consumer's electricity procurement cost need not reduce total system operating cost. We formalize this divergence through the notion of \emph{misalignment}.

\begin{definition}[Misalignment]
We say that misalignment occurs when the flexible consumer's cost-minimizing load shift increases total system operating cost relative to the no-shift baseline, i.e.
\begin{align}
    & V(\delta^*) > V(0),
    & \delta^* \in \argmin_{\delta \in \mathcal{F}(\alpha)} \, \Pi(\delta).
    \label{eq:misalignment}
\end{align}
\end{definition}

Misalignment captures the possibility that decentralized, price-anticipatory load shifting---while individually optimal for flexible consumers---induces higher cost generation at the system level relative to the no-shift baseline (rather than relative to the system-optimal shift). In such cases, flexible demand exploits its influence on prices rather than alleviating system constraints, leading to a loss of system operating efficiency. The remainder of this work characterizes when this divergence can and cannot arise, linking misalignment to operating regime transitions in the underlying DC-OPF.

\section{Theoretical Misalignment Results}\label{sec:theory}
In this section, we characterize when decentralized, price-anticipatory load shifting reduces system operating cost and when it does not.
Our analysis exploits the piecewise-linear structure of the DC-OPF, under which dispatch outcomes and LMPs depend on the set of binding transmission and generation constraints.

We analyze these outcomes through the lens of \emph{operating regimes}: regions of the load-shift space over which the DC-OPF active set remains fixed.
Within a regime, dispatch and prices vary affinely with load, and the mapping from load shifts to system costs and consumer payments is well behaved.
Misalignment can therefore arise only when load shifting induces transitions between operating regimes, triggering discrete changes in marginal generators or congestion patterns.

To formalize this structure, we define active constraint sets and the corresponding regions of load shifts for which these sets remain unchanged.

\begin{definition}[Active-set regions]
\label{def:active_set_regions}
An \emph{active set} is a tuple $\mathcal A=(\mathcal L^+,\mathcal L^-,\mathcal G^+,\mathcal G^-)$
specifying binding transmission and generation limits. For a load shift $\delta$, we say an optimal DC-OPF solution $(p,\theta,f)$ at load $d+\delta$ \emph{realizes} $\mathcal A$ if
\begin{align*}
    & f_\ell=\bar F_\ell, \; \forall \ell\in\mathcal L^+,
    & f_\ell=-\bar F_\ell, \; \forall \ell\in\mathcal L^-, \\
    & p_g=\bar p_g, \; \forall g\in\mathcal G^+,
    & p_g=0, \; \forall g\in\mathcal G^-.
\end{align*}
The \emph{active-set region} associated with $\mathcal A$ is then defined as
\begin{align*}
    \mathcal R(\mathcal A) \coloneqq \big\{\delta\in\mathcal F(\alpha)\ \big|\ \exists\ \text{an optimal $(p,\theta,f)$ realizing } \mathcal A\big\}.
\end{align*}
\end{definition}
Distinct active sets correspond to different operating regimes---defined by congestion patterns and generators at bounds---and the associated regions $\mathcal R(\mathcal A)$ are polyhedral, as they are characterized by linear feasibility and optimality conditions under a fixed binding pattern. In fact, related active-set representations of DC-OPF have already been used to predict dispatch and price outcomes \cite{ng2018statistical,ji2015probabilistic} and to analyze strategic generation investment \cite{taheri2021strategic}, highlighting their usefulness as a structural lens for market analysis.

\subsection{Alignment Under Limited Flexibility}
As long as marginal generators and congestion patterns---i.e., the set of binding constraints in \eqref{eq:dcopf}---remain unchanged, LMPs are locally invariant. In this case, reductions in flexible consumer procurement costs translate one-for-one into reductions in total system operating cost. Within a single active-set region, the flexible consumer therefore cannot exert market power over prices, and decentralized load shifting is necessarily aligned with the system operator objective. Proposition~\ref{prop:sufficient_condition_for_alignment} formalizes this intuition by providing a sufficient condition for alignment in terms of the feasible flexibility level $\alpha$ and available generation and transmission headroom.

\begin{proposition}\label{prop:sufficient_condition_for_alignment}
Suppose the flexible consumer’s feasible load-shifting set $\mathcal F(\alpha)$ is sufficiently small that, for all admissible load shifts $\delta \in \mathcal F(\alpha)$, the DC-OPF clears without changing
\begin{enumerate}
    \item the set of marginal generators, and
    \item the set of binding transmission constraints,
\end{enumerate}
relative to the baseline dispatch at $\delta = 0$. Equivalently, assume that
\begin{align*}
    \mathcal F(\alpha) \subset \mathcal R(\mathcal A_0),
\end{align*}
where $\mathcal A_0$ denotes the active set associated with the baseline operating point. Then minimizing the flexible consumer’s electricity procurement cost is equivalent to minimizing total system operating cost:
\begin{align*}
    \argmin_{\delta \in \mathcal F(\alpha)} V(\delta)
    \;=\;
    \argmin_{\delta \in \mathcal F(\alpha)} \Pi(\delta),
\end{align*}
and misalignment \eqref{eq:misalignment} cannot occur.
\end{proposition}
Proposition~\ref{prop:sufficient_condition_for_alignment} follows directly from the piecewise-affine structure of the DC-OPF (see Appendix~\ref{sec:sufficient_condition_for_alignment_proof} for the proof), but it plays an important conceptual role. It shows that misalignment cannot arise unless flexible demand is large and/or flexible enough to alter the system's operating regime---specifically, to change the set of binding transmission or generation constraints. When the feasible set of load shifts remains confined to $\mathcal{R}(\mathcal{A}_0)$, the mapping from load to dispatch and prices is affine, and minimizing consumer procurement cost is equivalent to minimizing system operating cost.

\subsection{Misalignment Requires Market Power}

When the feasible set of load shifts is sufficiently large, however, a flexible consumer can exert market power by strategically positioning demand so as to influence which generators or transmission constraints are marginal in the market clearing solution. In particular, a consumer might reallocate load to avoid triggering a high cost generator at one location, thereby lowering the prices it faces locally. However, this same reallocation can force the system operator to rely more heavily on even higher cost resources that are already marginal elsewhere in the network.

It is this form of strategic gaming---in which load shifting suppresses one marginal resource while exacerbating reliance on others---that creates the potential for misalignment between private procurement costs and total system operating cost. Proposition~\ref{prop:sufficient_condition_for_alignment} showed that such behavior cannot arise unless load shifting alters the set of binding constraints. We now make this mechanism precise by characterizing misalignment in terms of transitions between operating regimes of the DC-OPF.

\begin{proposition}[Misalignment occurs at intersections of active-set regions]
\label{prop:necessary_condition_for_misalignment}
Suppose there exists an optimal shift $\delta' \in \argmin_{\delta\in\mathcal F(\alpha)} \Pi(\delta)$ such that $V(\delta')>V(0)$. Then, there exists an optimal $\delta^* \in \argmin_{\delta\in\mathcal F(\alpha)} \Pi(\delta)$ with $V(\delta^*)>V(0)$ and two distinct active sets $\mathcal A_1\neq \mathcal A_2$ such that
\begin{align*}
\delta^* \in \mathcal R(\mathcal A_1)\cap \mathcal R(\mathcal A_2).
\end{align*}
Equivalently, misalignment can only occur if the optimal solution to \eqref{eq:load_shifting} positions the system at a boundary between active sets where the identity of marginal generators and/or congested transmission lines changes.
\end{proposition}

Proposition~\ref{prop:necessary_condition_for_misalignment} establishes a central economic conclusion of this paper. Misalignment cannot arise from a load shift that lies in the interior of a DC-OPF operating regime: although such shifts may increase total system operating cost, they are never optimal for the flexible consumer. Rather, misalignment can occur only when the consumer’s cost-minimizing load shift places the system at the intersection of multiple active-set regions. At such regime boundaries, arbitrarily small changes in net load can switch which generators are marginal or which transmission constraints bind, inducing instantaneous and non-smooth changes in LMPs. These price discontinuities decouple the flexible consumer’s procurement incentives from total system operating cost, allowing a load shift that lowers consumer payments while increasing total generation cost.

\section{Three-Zone Illustrative Example}
\subsection{System Description}
To illustrate the theoretical results of Section~\ref{sec:theory}, we consider a stylized three-zone system designed to make operating regime transitions and price endogeneity explicit. The network consists of three zones (A, B, and C) connected by (uniform susceptance) capacity-limited transmission links, giving rise to multiple operating regimes. Generation costs differ across zones, with high cost resources in Zone~A, intermediate-cost resources in Zone~B, and low cost generation in Zone~C. Load is distributed asymmetrically across the network, with flexible demand located in all zones but concentrated in Zone~C; full parameter values are provided in Fig.~\ref{fig:three_bus_system}. This configuration creates opposing consumer incentives to either shift load toward lower cost generation or to strategically alter marginal generators or congestion patterns, allowing decentralized load shifting to diverge from system cost minimization.

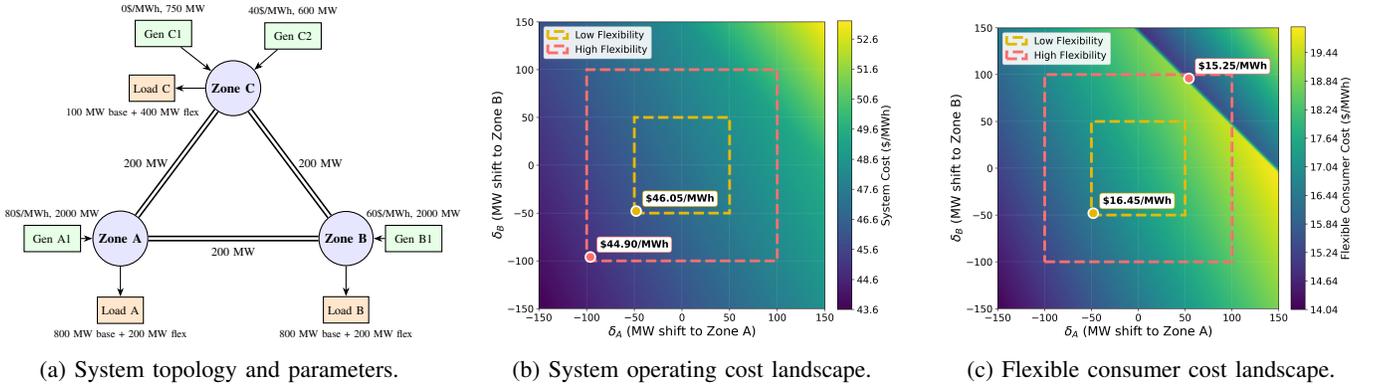
\begin{figure*}[ht]
     \centering
     \begin{subfigure}[b]{0.32\textwidth}
         \centering
         \scalebox{0.5}{\begin{tikzpicture}[
    bus/.style={circle, draw=black, thick, minimum size=1.2cm, fill=blue!10},
    gen/.style={rectangle, draw=black, thick, minimum width=1.5cm, minimum height=0.7cm, fill=green!10},
    load/.style={rectangle, draw=black, thick, minimum width=1.2cm, minimum height=0.7cm, fill=orange!20},
    line/.style={thick, -{Stealth}},
    label/.style={font=\small}
]

\node[bus] (bus0) at (0, 0) {\textbf{Zone A}};
\node[bus] (bus1) at (6, 0) {\textbf{Zone B}};
\node[bus] (bus2) at (3, 4) {\textbf{Zone C}};

\node[gen, left=0.3cm of bus0, label] (genA1) {Gen A1};
\node[label, above=0.05cm of genA1] {\footnotesize 80\$/MWh, 2000 MW};

\node[gen, right=0.3cm of bus1, label] (genB1) {Gen B1};
\node[label, above=0.05cm of genB1] {\footnotesize 60\$/MWh, 2000 MW};

\node[gen, above left=0.8cm of bus2, label] (genC1) {Gen C1};
\node[label, above=0.05cm of genC1] {\footnotesize 0\$/MWh, 750 MW};
\node[gen, above right=0.5cm and 0.3cm of bus2, label] (genC2) {Gen C2};
\node[label, above=0.05cm of genC2] {\footnotesize 40\$/MWh, 600 MW};

\node[load, below=0.8cm of bus0, label] (load0) {Load A};
\node[label, below=0.05cm of load0] {\footnotesize 800 MW base + 200 MW flex};

\node[load, below=0.8cm of bus1, label] (load1) {Load B};
\node[label, below=0.05cm of load1] {\footnotesize 800 MW base + 200 MW flex};

\node[load, left=0.8cm of bus2, label] (load2) {Load C};
\node[label, below left=0.05cm and -2cm of load2] {\footnotesize 100 MW base + 400 MW flex};

\draw[line] (genA1) -- (bus0);
\draw[line] (genB1) -- (bus1);
\draw[line] (genC1) -- (bus2);
\draw[line] (genC2) -- (bus2);

\draw[line] (bus0) -- (load0);
\draw[line] (bus1) -- (load1);
\draw[line] (bus2) -- (load2);

\draw[very thick, double distance=2pt] (bus0) -- (bus1) 
    node[midway, below=0.1cm, label] {200 MW};
    
\draw[very thick, double distance=2pt] (bus0) -- (bus2) 
    node[midway, left=0.1cm, label] {200 MW};
    
\draw[very thick, double distance=2pt] (bus1) -- (bus2) 
    node[midway, right=0.1cm, label] {200 MW};

\end{tikzpicture}}
         \caption{System topology and parameters.}
         \label{fig:three_bus_system}
     \end{subfigure}
     \hfill
     \begin{subfigure}[b]{0.3\textwidth}
         \centering
         \includegraphics[width=\textwidth]{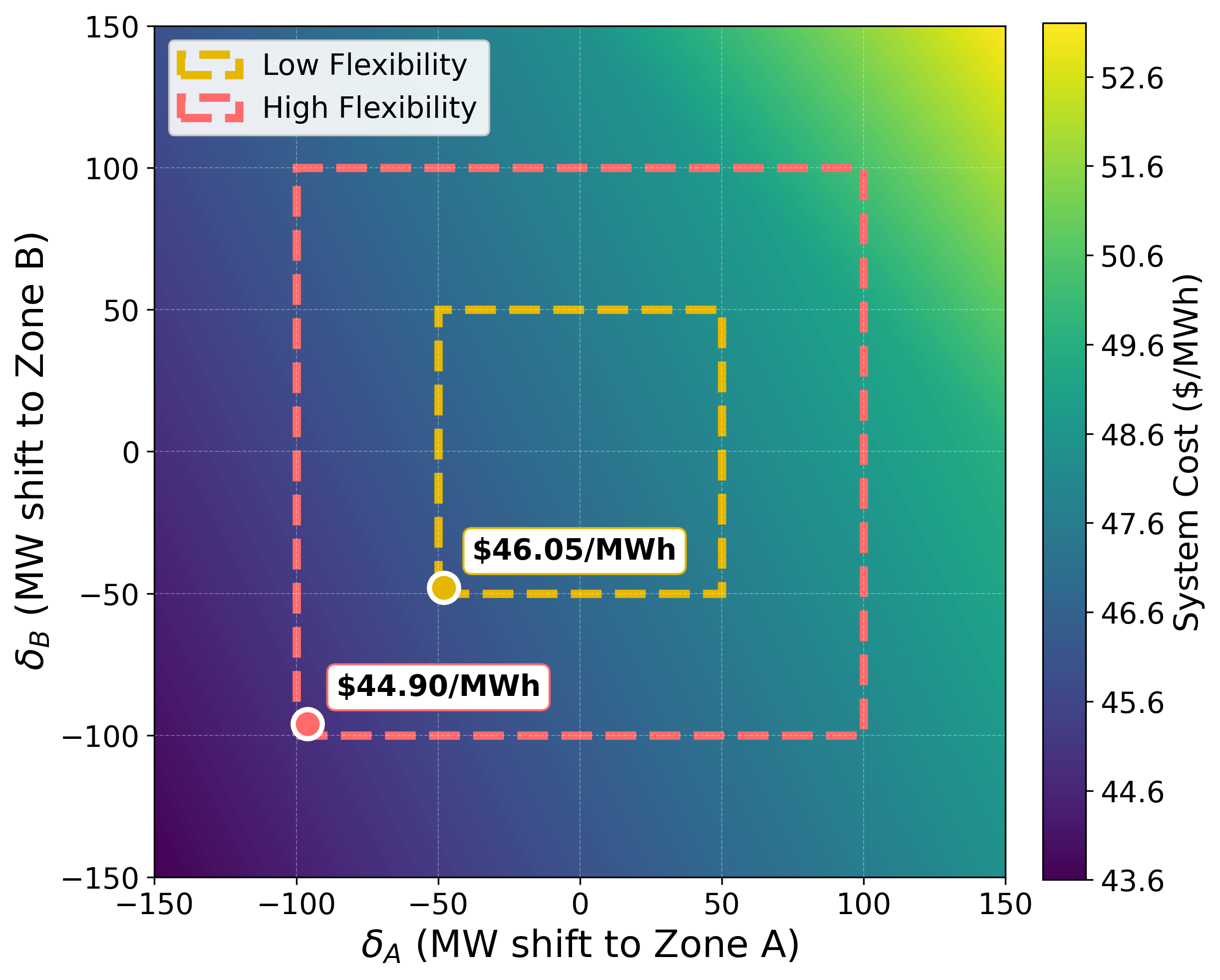}
         \caption{System operating cost landscape.}
         \label{fig:system_cost_landscape}
     \end{subfigure}
     \hfill
     \begin{subfigure}[b]{0.3\textwidth}
         \centering
         \includegraphics[width=\textwidth]{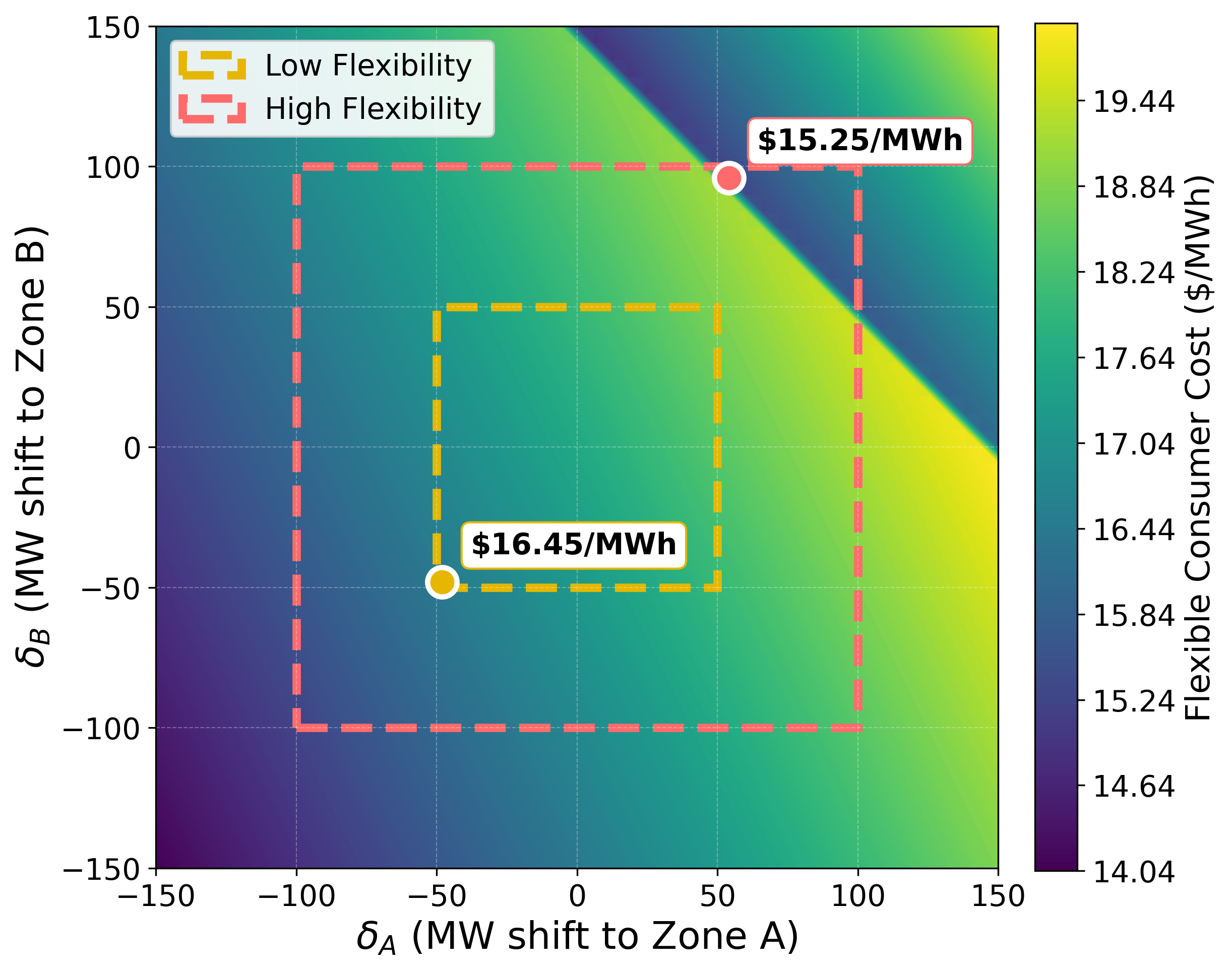}
         \caption{Flexible consumer cost landscape.}
         \label{fig:consumer_cost_landscape}
     \end{subfigure}
     \label{fig:three_bus}
     \caption{Three-zone illustrative example. Misalignment arises from price discontinuities at active-set region boundaries; compare the discontinuous flexible consumer objective in Fig.~\ref{fig:consumer_cost_landscape} with the continuous, convex system operating cost in Fig.~\ref{fig:system_cost_landscape}.}
\end{figure*}

\subsection{Impacts of Load Shifting}

We evaluate the system operating cost $V(\delta)$ and the flexible consumer procurement cost $\Pi(\delta)$ over the range $(\delta_A,\delta_B)\in[-150,150]^2$, with $\delta_C=-(\delta_A+\delta_B)$ enforcing the box-with-balance shifting constraints (see Eq.~\eqref{eq:box_with_balance}). Figures~\ref{fig:system_cost_landscape} and~\ref{fig:consumer_cost_landscape} visualize the resulting cost landscapes, while Table~\ref{tab:cost_comp} compares the optimal outcomes under centralized and decentralized load-shifting objectives (normalized to units of \$/MWh) for varying levels of flexibility~$\alpha$.

\begin{table}[h]
\centering
\caption{Three-Zone Centralized vs. Decentralized Shifting}
\label{tab:cost_comp}
\begin{tabular}{lllccc}
\toprule
$\boldsymbol{\alpha}$ & \textbf{Objective} & \textbf{Load Shift} $\boldsymbol{\delta}$ & $\boldsymbol{V}$ & $\boldsymbol{\Pi}$ & $\boldsymbol{\Delta V}$ \\
\midrule
Baseline & --- & (0, 0, 0) & 47.20 & 17.60 & 0.00 \\
\midrule
Low Flex. & Centralized & (-48, -48, 96) & 46.05 & 16.45 & 1.15 $\downarrow$ \\
 & Decentralized & (-48, -48, 96) & 46.05 & 16.45 & 1.15 $\downarrow$ \\
High Flex. & Centralized & (-96, -96, 192) & 44.90 & 15.32 & 2.30 $\downarrow$ \\
 & Decentralized & (54, 96, -150) & 48.83 & 15.25 & 1.63 $\uparrow$ \\
\bottomrule
\end{tabular}
\begin{tablenotes}
\footnotesize
\item \textit{Notes:} Cost values are reported in \$/MWh.
$V$ denotes system operating cost, $\Pi$ denotes the flexible consumer cost, and $\Delta V$ reports the change in system operating cost relative to the no-shift baseline ($\alpha = 0$).
\end{tablenotes}
\end{table}

\textbf{Alignment under limited flexibility (Proposition~\ref{prop:sufficient_condition_for_alignment}).}
Under low flexibility ($\alpha = 0.25$), the load shift selected by the flexible consumer coincides with the load shift that minimizes total system operating cost (Table~\ref{tab:cost_comp}), i.e.,
\begin{align*}
    \argmin_{\delta\in\mathcal{F}(0.25)} V(\delta)
    \;=\;
    \argmin_{\delta\in\mathcal{F}(0.25)} \Pi(\delta).
\end{align*}
This outcome arises because the feasible load-shifting set is confined to a subset of the baseline active-set region $\mathcal{R}(\mathcal{A}_0)$. Within this region, the set of binding transmission and generation constraints remains unchanged, so dispatch varies affinely with load and LMPs are constant over $\mathcal{F}(0.25)$. Consequently, reductions in the flexible consumer’s procurement cost translate one-for-one into reductions in system operating cost, and decentralized load shifting is necessarily aligned with the system objective.

\textbf{Misalignment requires market power (Proposition~\ref{prop:necessary_condition_for_misalignment}).}
Under higher flexibility ($\alpha = 0.5$), the feasible set expands beyond the baseline active-set region, permitting load shifts that cross congestion and marginal-generator boundaries. In this regime, the sufficient condition for alignment in Proposition~\ref{prop:sufficient_condition_for_alignment} no longer holds, and misalignment emerges as the consumer-optimal load shift diverges from the system-optimal solution.

In particular, the consumer-optimal outcome involves a strategic reallocation of load toward higher-LMP zones (A and B) that displaces generator C2 from the marginal set. This operating regime transition lowers the LMP at Zone~C from \$40/MWh to \$0/MWh, reducing the flexible consumer’s procurement cost by \$2.35/MWh (13.3\%). However, the same shift forces the system to rely more heavily on higher cost generation elsewhere, increasing system operating cost by \$1.63/MWh (3.4\%). This three-zone example therefore illustrates that misalignment arises not from gradual inefficiencies in redispatch, but from discrete operating regime transitions induced by strategic load positioning, consistent with the mechanism characterized in Proposition~\ref{prop:necessary_condition_for_misalignment}.

\section{Evaluation on the RTS-GMLC Test System}\label{sec:computational}
\subsection{Experimental Setup}
To assess how the theoretical mechanisms identified in Section~\ref{sec:theory} manifest in a realistic, large-scale setting, we evaluate strategic spatial load shifting on the synthetic 73-bus IEEE RTS-GMLC test system \cite{rts}. This system captures realistic network congestion and generation heterogeneity, allowing us to study how price-anticipatory load shifting affects flexible consumer costs and total system operating cost. We additionally track changes in inflexible consumer costs and generator profits to characterize surplus redistribution across market participants.

Our analysis considers hourly market clearing and spatial load shifting for all 8{,}784 hours in the RTS-GMLC dataset. We solve a single-level reformulation of Problem~\eqref{eq:load_shifting} (see Problem~\eqref{eq:reformulation} in Appendix~\ref{app:reformulation}) independently for each hour using Gurobi~12 \cite{gurobi}. Following \cite{GorkaEtal2024}, we model four large spatially flexible loads located at buses 103, 107, 204, and 322, each with nominal demand of 250~MW, together accounting for approximately 20\% of total energy consumption over the study period. To ensure feasibility in all hours, generator capacities are uniformly scaled by a factor of 1.25.

We compare two load-shifting formulations: (i) a decentralized, price-anticipatory flexible consumer problem \eqref{eq:load_shifting}, and (ii) a centralized benchmark in which the flexible consumer objective is replaced by the system operator objective \eqref{eq:dcopf_new}, corresponding to a central planner solution. We consider two flexibility levels, $\alpha \in \{0.25, 0.5\}$. For each setting, results are reported relative to a no-flexibility baseline ($\delta = 0$) in terms of changes in system operating cost, flexible and inflexible consumer procurement costs, and generator profits.

\subsection{Cost and Profit Outcomes}
Table~\ref{tab:cost_savings} reports annualized changes in system operating costs, consumer procurement costs, and generator profits across all experimental scenarios; for a visual reference, see Fig.~\ref{fig:cost_savings} in Appendix~\ref{app:figures}, which shows the corresponding distributions of hourly outcomes as a boxplot. Together, these results reveal how decentralized, price-anticipatory load shifting redistributes surplus from generators to consumers and when such redistribution aligns (or fails to align) with system objectives. We highlight key findings below.

\begin{table*}[htbp]
\centering
\caption{Cost and Profit Changes under Load Shifting}
\label{tab:cost_savings}
\setlength{\tabcolsep}{6pt}
\begin{tabular}{llcccccccccc}
\toprule
& & 
\multicolumn{2}{c}{\textbf{Flexible Consumer}} & 
\multicolumn{2}{c}{\textbf{Inflexible Consumer}} & 
\multicolumn{2}{c}{\textbf{Generator}} & 
\multicolumn{3}{c}{\textbf{System Operator}} \\

& & \multicolumn{2}{c}{\footnotesize Base Cost: \$242.9M}
& \multicolumn{2}{c}{\footnotesize Base Cost: \$918.8M}
& \multicolumn{2}{c}{\footnotesize Base Profit: \$510.4M}
& \multicolumn{3}{c}{\footnotesize Base Cost: \$548.6M} \\

\cmidrule(lr){3-4} \cmidrule(lr){5-6} \cmidrule(lr){7-8} \cmidrule(lr){9-11}

\textbf{Objective} & $\boldsymbol{\alpha}$ 
& \textbf{\$M} & \textbf{\%} 
& \textbf{\$M} & \textbf{\%} 
& \textbf{\$M} & \textbf{\%} 
& \textbf{\$M} & \textbf{\%} & \textbf{Misalign \%} \\

\midrule
\multirow{2}{*}{Flexible Consumer}
& 0.25 & 28.2 $\downarrow$ & 11.6\% $\downarrow$ & 17.7 $\downarrow$ & 1.9\% $\downarrow$ & 14.9 $\downarrow$ & 2.9\% $\downarrow$ & 3.1 $\downarrow$ & 0.6\% $\downarrow$ & 4.6\% \\
& 0.50 & \textbf{30.6 $\downarrow$} & \textbf{12.6\%} $\downarrow$ & \textbf{23.0} $\downarrow$ & \textbf{2.5\%} $\downarrow$ & 20.1 $\downarrow$ & 3.9\% $\downarrow$ & 4.2 $\downarrow$ & 0.8\% $\downarrow$ & 10.7\% \\

\midrule
\multirow{2}{*}{System Operating Cost}
& 0.25 & 24.8 $\downarrow$ & 10.2\% $\downarrow$ & 4.9 $\downarrow$ & 0.5\% $\downarrow$ & 0.9 $\uparrow$ & 0.2\% $\uparrow$ & 3.4 $\downarrow$ & 0.6\% $\downarrow$ & -- \\
& 0.50 & 23.1 $\downarrow$ & 9.5\% $\downarrow$ & 3.4 $\downarrow$ & 0.4\% $\downarrow$ & \textbf{5.6} $\uparrow$ & \textbf{1.1\%} $\uparrow$ & \textbf{4.7} $\downarrow$ & \textbf{0.9\%} $\downarrow$ & -- \\

\bottomrule
\end{tabular}

\begin{tablenotes}
\footnotesize
\item Notes: Results aggregated over 8,784 hours (366 days $\times$ 24 hours). 
Monetary values (denoted by \$M) are in millions of USD per year. Percent change from baseline is shown in columns labeled \%. 
Misalign \% shows the percentage of hours where flexible consumer optimization leads to increased system operating costs. 
\textbf{Bold values} indicate the most favorable outcome for each stakeholder across all settings.
\end{tablenotes}
\end{table*}

\textbf{Flexible consumer benefits from strategic shifting.}
Table~\ref{tab:cost_savings} shows that the flexible consumer achieves substantial procurement cost reductions under decentralized load shifting---up to \$30.6~million (12.6\% of baseline costs) over the one-year period for $\alpha=0.5$. Relative to system-optimal shifting, price-anticipatory behavior yields an additional \$7.5~million in savings (3.1\%), indicating clear incentives for flexible consumers to act strategically rather than to coordinate with system-wide cost minimization.

\textbf{Inflexible consumers receive modest spillover benefits.} Although strategic load shifting can, in principle, raise prices for inflexible consumers by altering marginal generators, this effect is dominated in our experiments by offsetting benefits. For $\alpha=0.5$, inflexible consumers experience average cost savings of up to 2.5\% of baseline costs under decentralized shifting, compared to 0.4\% under system-optimal shifting, suggesting that shared incentives to displace high cost generation often outweigh adverse redispatch effects.

\textbf{Generator profits fall under decentralized shifting but rise under system-optimal shifting.} When the flexible consumer minimizes procurement costs, generator profits decline modestly (3.9\% for $\alpha=0.5$), reflecting downward pressure on LMPs. In contrast, system-optimal shifting increases generator profits by 1.1\%. Ultimately, we find in our experiments that system-optimal shifting benefits all stakeholder groups to some degree, while decentralized shifting redistributes surplus from generators to consumers.

\textbf{Aggregate system cost impacts are small and similar across objectives.} Despite these distributional effects, aggregate system operating cost savings remain modest under decentralized load shifting (approximately 0.8\% of baseline costs for $\alpha=0.5$), with system-optimal shifting improving savings by only an additional 0.1\%. This indicates that the primary effects of strategic load shifting are redistributive rather than strongly inefficient at the system level.

\textbf{Flexibility enables market power and misalignment.} Increasing flexibility expands the feasible load-shifting set, enabling the flexible consumer to influence marginal generators and congestion patterns. Consistent with Proposition~\ref{prop:sufficient_condition_for_alignment}, raising $\alpha$ from $0.25$ to $0.5$ increases the incidence of misalignment from 4.6\% to 10.7\% of all hours, highlighting that misalignment arises through operating regime transitions rather than small adjustments within a fixed regime.

\subsection{Strategic Regime Shifting}
Fig.~\ref{fig:merit_order} summarizes how price-anticipatory spatial load shifting alters generator marginal status and dispatch outcomes along the active portion of the merit order. The left panels report the net change in the frequency with which each generator is marginal across all hours, where positive (red) values indicate that a generator becomes marginal more often and negative (blue) values indicate that it becomes marginal less often. The right panels report the corresponding net change in dispatched energy. Cases in which decentralized load shifting increases total system operating cost (misalignment) are shown in Fig.~\ref{fig:misaligned_merit_order}, while cases in which load shifting reduces system operating cost (alignment) are shown in Fig.~\ref{fig:aligned_merit_order}.

\begin{figure}[htbp]
    \centering

    \begin{subfigure}[t]{\columnwidth}
        \centering
        \includegraphics[width=\textwidth]{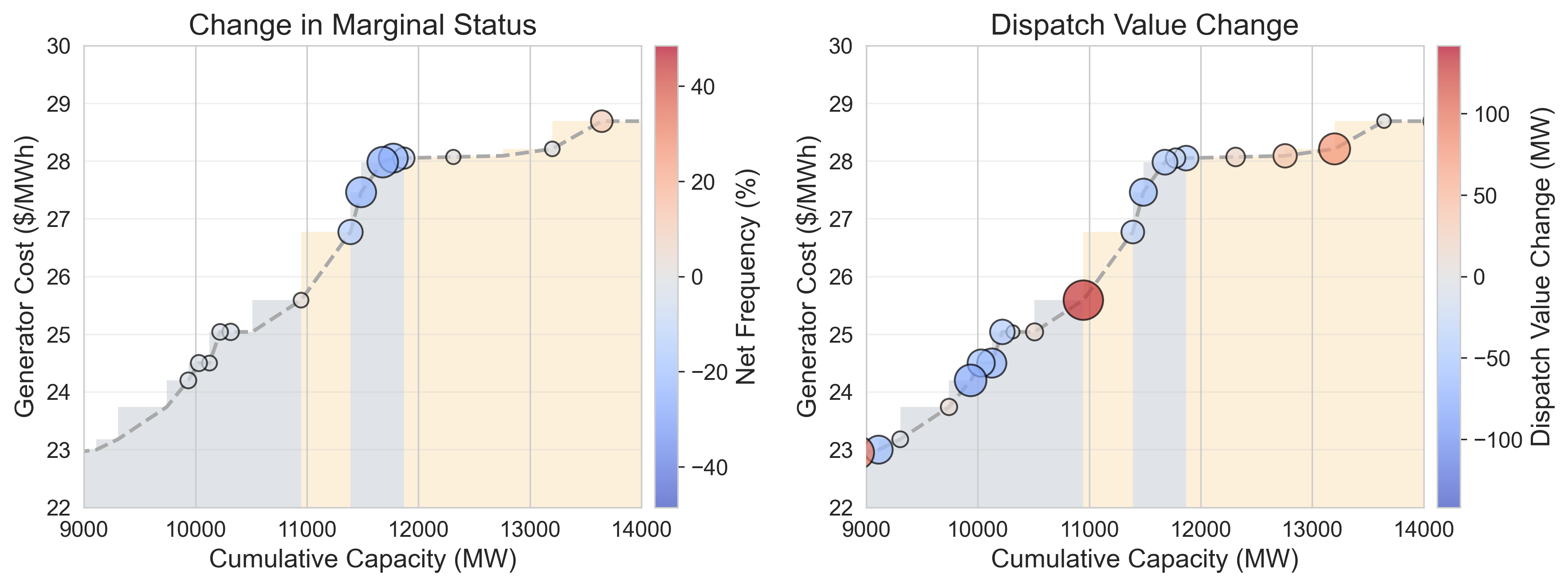}
        \caption{Misaligned instances}
        \label{fig:misaligned_merit_order}
    \end{subfigure}

    \vspace{0.8em}

    \begin{subfigure}[t]{\columnwidth}
        \centering
        \includegraphics[width=\textwidth]{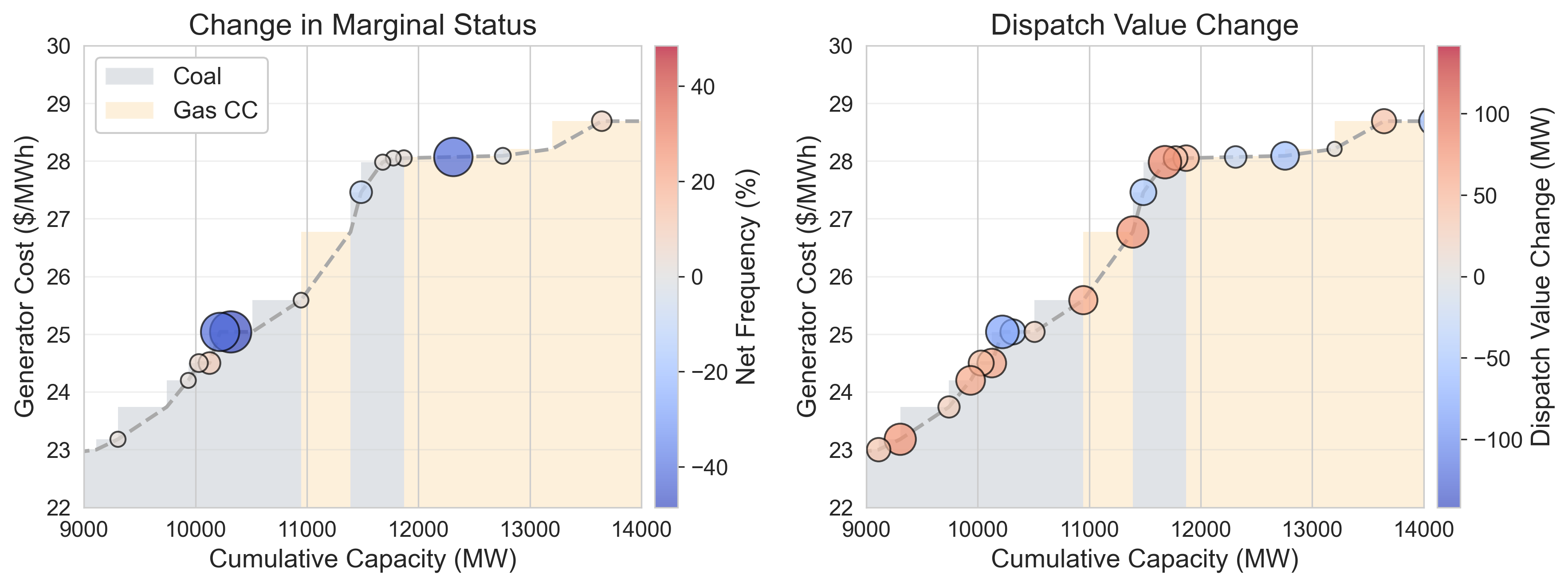}
        \caption{Aligned instances}
        \label{fig:aligned_merit_order}
    \end{subfigure}

    \caption{Impacts of decentralized shifting on generator marginal status and dispatch outcomes over a segment of the merit order.}
    \label{fig:merit_order}
\end{figure}

Proposition~\ref{prop:necessary_condition_for_misalignment} states that adverse system-level outcomes of price-anticipatory load shifting arise through interactions with discontinuities in the generation merit order, rather than through smooth redispatch within a fixed operating regime. Figure~\ref{fig:misaligned_merit_order} illustrates this mechanism. In misaligned instances, load shifting concentrates near steep ``cliffs'' in the merit order, where small changes in net load induce discrete changes in marginal generators and prices. Strategic reallocation suppresses generators near these discontinuities from the marginal set while increasing reliance on higher cost generators that are already marginal elsewhere, reducing flexible consumer procurement costs while increasing total system operating cost. By contrast, aligned instances (Fig.~\ref{fig:aligned_merit_order}) are characterized by redispatch that occurs ``downhill'' along the merit order, so that displaced generation is replaced by lower cost resources and total system operating cost falls.

\section{Conclusion}
This paper examines the system-level impacts of price-anticipatory spatial load shifting by large flexible consumers, modeling the interaction between a flexible consumer and the system operator as a Stackelberg game. We show theoretically that market power is necessary for misalignment: consumer and system objectives are aligned whenever feasible load shifts stay within a single operating regime, and misalignment arises only when shifting changes the set of congested lines or marginal generators. Computational results on the RTS-GMLC test system support these insights: decentralized shifting delivers aggregate system cost savings comparable to a centralized benchmark, while misalignment occurs in a subset of hours driven by operating regime transitions. These effects are primarily redistributive, shifting surplus across market participants with relatively small impacts on aggregate system efficiency.

Future work will extend these results to joint spatial and temporal flexibility for multi-period shifting, emissions-based shifting---which requires characterizing how endogenous dispatch and prices respond to emissions signals across operating regimes---and to settings with multiple large flexible consumers, where interaction may induce alignment or compound misalignment.

\section*{Acknowledgements}
This work was supported by the Future Energy Systems Center through the MIT Energy Initiative.

\section*{AI Usage Disclosure}
Large language models were used as a support tool for drafting and editing text and for coding-related tasks.

\bibliographystyle{ieeetr}
\bibliography{bibliography}

\appendix

\subsection{Proof of Proposition~\ref{prop:sufficient_condition_for_alignment}}\label{sec:sufficient_condition_for_alignment_proof}
\begin{proof}
Since $\mathcal{F}(\alpha) \subset \mathcal{R}(\mathcal{A}_0)$, the LMP vector is constant throughout $\mathcal{F}(\alpha)$; denote this constant vector by $\lambda$. Uniqueness of the LMP vector is not required here: within $\mathcal{R}(\mathcal{A}_0)$, all dual optimal solutions associated with the nodal balance constraints remain constant with respect to $\delta$. By the envelope theorem, the gradient of the system operating cost $V$ with respect to the load shift $\delta$ satisfies
\begin{align*}
\nabla_\delta V(\delta) = \frac{\partial \mathcal{L}}{\partial \delta} = \lambda, \qquad \forall \delta\in\mathcal{F}(\alpha),
\end{align*}
where the last equality follows from the fact that $\delta$ enters the DC-OPF only via the balance constraint \eqref{eq:energy_balance}.

The procurement cost is given by $\Pi(\delta) = \lambda^\top (d_\mathrm{flex} + \delta)$, and hence $\nabla_\delta \Pi(\delta) = \lambda$. Define $h(\delta) \coloneqq V(\delta) - \Pi(\delta)$. Then $\nabla h(\delta) = 0$ for all $\delta \in \mathcal{F}(\alpha)$, implying that $h$ is constant on $\mathcal{F}(\alpha)$. Therefore, there exists a constant $k$ such that
\begin{align*}
V(\delta) = \Pi(\delta) + k \qquad \forall \delta\in\mathcal{F}(\alpha).
\end{align*}
Consequently,
\begin{align*}
\argmin_{\delta\in\mathcal{F}(\alpha)} V(\delta) = \argmin_{\delta\in\mathcal{F}(\alpha)} \Pi(\delta),
\end{align*}
which establishes the claim.
\end{proof}

\subsection{Proof of Proposition~\ref{prop:necessary_condition_for_misalignment}}
To prove Proposition~\ref{prop:necessary_condition_for_misalignment}, it is helpful to first observe that, for misalignment to occur, the optimal load shift must also be misaligned with respect to the LMPs for the corresponding active-set region.
\begin{lemma}[Dual subgradient inequality]\label{lem:dual_subgradient}
Let $\delta^*$ satisfy $V(\delta^*)>V(0)$.
Let $\lambda$ be any optimal dual variable associated with the DC-OPF solution under load shift $\delta^*$
(equivalently, any $\lambda\in \partial_\delta V(\delta^*)$).
Then $\lambda^\top \delta^*>0$.
\end{lemma}

\begin{proof}[Proof of Lemma~\ref{lem:dual_subgradient}]
Since $V$ is the value function of a linear program with respect to a right-hand-side perturbation,
it is convex in $\delta$. Hence for any $\delta$ and any $\lambda\in \partial_\delta V(\delta^*)$,
\begin{align*}
V(\delta)\;\ge\; V(\delta^*) + \lambda^\top(\delta-\delta^*).
\end{align*}
Choosing $\delta=0$ gives
\begin{align*}
V(0)\;\ge\; V(\delta^*) - \lambda^\top \delta^*,
\end{align*}
so $\lambda^\top\delta^* \ge V(\delta^*)-V(0)>0$.
\end{proof}

\begin{proof}[Proof of Proposition~\ref{prop:necessary_condition_for_misalignment}]
Let $\delta^*$ be an optimal solution of the flexible consumer problem and suppose $V(\delta^*)>V(0)$.
Let $\mathcal R^*$ denote the (single) active-set region containing $\delta^*$ and let $\lambda$ be the (constant)
LMP vector on $\mathcal R^*$.

Consider the linear program obtained by fixing prices to $\lambda$ and restricting $\delta$ to $\mathcal R^*$:
\begin{subequations}\label{eq:critical_region_shifting}
\begin{align}
    \min_\delta \quad & \lambda^{\top} (d_\mathrm{flex} + \delta) \\
    \mathrm{s.t.} \quad & T(\alpha)\delta \leq q(\alpha), \label{eq:critical_region_flexibility} \\
    & 1^\top \delta = 0, \\
    & A_{\mathcal{R}^*} \delta \leq b_{\mathcal{R}^*}, \label{eq:critical_region_bounds}
\end{align}
\end{subequations}
where \eqref{eq:critical_region_bounds} defines $\mathcal R^*$, without loss of generality, as inequalities. Since $\delta^*\in \mathcal R^*$ and $\lambda$ is constant on $\mathcal R^*$, $\delta^*$ must satisfy the KKT optimality conditions of \eqref{eq:critical_region_shifting}.

Let $\gamma, \rho\ge 0$ denote the dual variables for \eqref{eq:critical_region_flexibility} and \eqref{eq:critical_region_bounds}, and let $\kappa\in\mathbb R$ denote the dual variable for $1^\top\delta=0$. The stationarity condition is given by
\begin{align}
\lambda + T(\alpha)^\top\gamma + \kappa 1 + A_{\mathcal R^*}^\top\rho = 0. \label{eq:stationarity}
\end{align}
If $\delta^*$ does not lie on the boundary between two or more active-set regions, then it lies in the interior of $\mathcal R^*$,
so $A_{\mathcal R^*}\delta^* - b_{\mathcal R^*}<0$ and complementary slackness implies $\rho=0$.
Thus \eqref{eq:stationarity} becomes
\begin{align*}
\lambda = -T(\alpha)^\top\gamma - \kappa 1.
\end{align*}
Taking the inner product with $\delta^*$ and using $1^\top\delta^*=0$ yields
\begin{align*}
\lambda^\top\delta^* = -\gamma^\top T(\alpha)\delta^*.
\end{align*}
By complementary slackness for constraint \eqref{eq:critical_region_flexibility},
\begin{align*}
\gamma^\top(T(\alpha)\delta^* - q(\alpha))=0
\quad\implies\quad
\gamma^\top T(\alpha)\delta^* = \gamma^\top q(\alpha),
\end{align*}
so
\begin{align*}
\lambda^\top\delta^* = -\gamma^\top q(\alpha)\le 0,
\end{align*}
since $\gamma\ge 0$ by dual feasibility and $q(\alpha)\ge 0$ by assumption. This contradicts Lemma~\ref{lem:dual_subgradient}, which implies
$\lambda^\top\delta^*>0$ when $V(\delta^*)>V(0)$.
Therefore, $\delta^*$ must lie on the boundary between two or more operating regimes.
\end{proof}

\subsection{Single-level Reformulation of \eqref{eq:load_shifting}}\label{app:reformulation}
The single-level reformulation of \eqref{eq:load_shifting} is given as follows:
\begin{subequations}\label{eq:reformulation}
\begin{align}
    \min \quad 
    & \lambda^\top\big(d_\mathrm{flex} + \delta\big) 
    \label{eq:reform_obj} \\
    \mathrm{s.t.} \quad 
    & (p,\theta,f) \ \text{satisfy \eqref{eq:flow_angles}--\eqref{eq:dcopf_refbus}} \label{eq:reform_primal} \\
    & -\lambda^\top G + \pi^+ - \pi^- = c & \label{eq:reform_stat_p} \\
    & -\eta^\top (BA) + \nu e_s^\top = 0 \label{eq:reform_stat_theta} \\
    & \eta^\top + \lambda^\top A^\top + \mu^+ - \mu^- = 0 \label{eq:reform_stat_f} \\
    & \mu^+, \mu^-, \pi^+, \pi^- \geq 0 & \label{eq:reform_dual_feas} \\
    & c^\top p = -\lambda^\top (d+\delta) + (\mu^+ - \mu^-)^\top \bar{F} + (\pi^+)^\top \bar{p} & \label{eq:reform_strong_duality} \\
    & \delta \in \mathcal{F}(\alpha) \label{eq:reform_delta_feas}.
\end{align}
\end{subequations}

Constraints~\eqref{eq:reform_primal} enforce primal feasibility of the market clearing problem, while \eqref{eq:reform_stat_p}--\eqref{eq:reform_dual_feas} impose dual feasibility. Strong duality is enforced via \eqref{eq:reform_strong_duality}, ensuring that the primal dispatch and dual price variables are mutually optimal for the realized load vector $d + \delta$. The upper-level decision variable $\delta$ enters the problem through the load balance constraint~\eqref{eq:energy_balance}, the strong duality constraint~\eqref{eq:reform_strong_duality}, and the flexible feasibility set~\eqref{eq:reform_delta_feas}, and the objective~\eqref{eq:reform_obj}. As a result, Problem~\eqref{eq:reformulation} is \emph{nonconvex} due to the bilinear terms \eqref{eq:reform_obj} and \eqref{eq:reform_strong_duality}. We solve the resulting bilinear program using Gurobi~12’s nonconvex quadratic solver, which employs a spatial branch-and-bound framework with convex relaxations \cite{gurobi}.

\subsection{RTS-GMLC Computational Details}
For linear offer curves, we utilize each generator's average heat rate (\verb|HR_avg_0|) from the RTS-GMLC data to compute a constant marginal cost in \$/MWh as:
\begin{equation}
c_g = \frac{p_g^{\text{fuel}}}{10^6} \times 10^3 \times \text{HR}_g
\end{equation}
where $p_g^{\text{fuel}}$ is the fuel price in \$/MMBTU and $\text{HR}_g$ is the average heat rate for generator $g$. This formulation yields a single linear segment for each generator's offer curve rather than the multi-segment piecewise linear curves available in the full RTS-GMLC dataset. For VRE generation resources, we apply time-varying upper bounds using the 5-minute resolution timeseries data provided with RTS-GMLC, aggregated to hourly resolution by averaging.

\subsection{Cost and Profit Outcomes Boxplot}\label{app:figures}

\begin{figure}[H]
    \centering
    \includegraphics[width=\columnwidth]{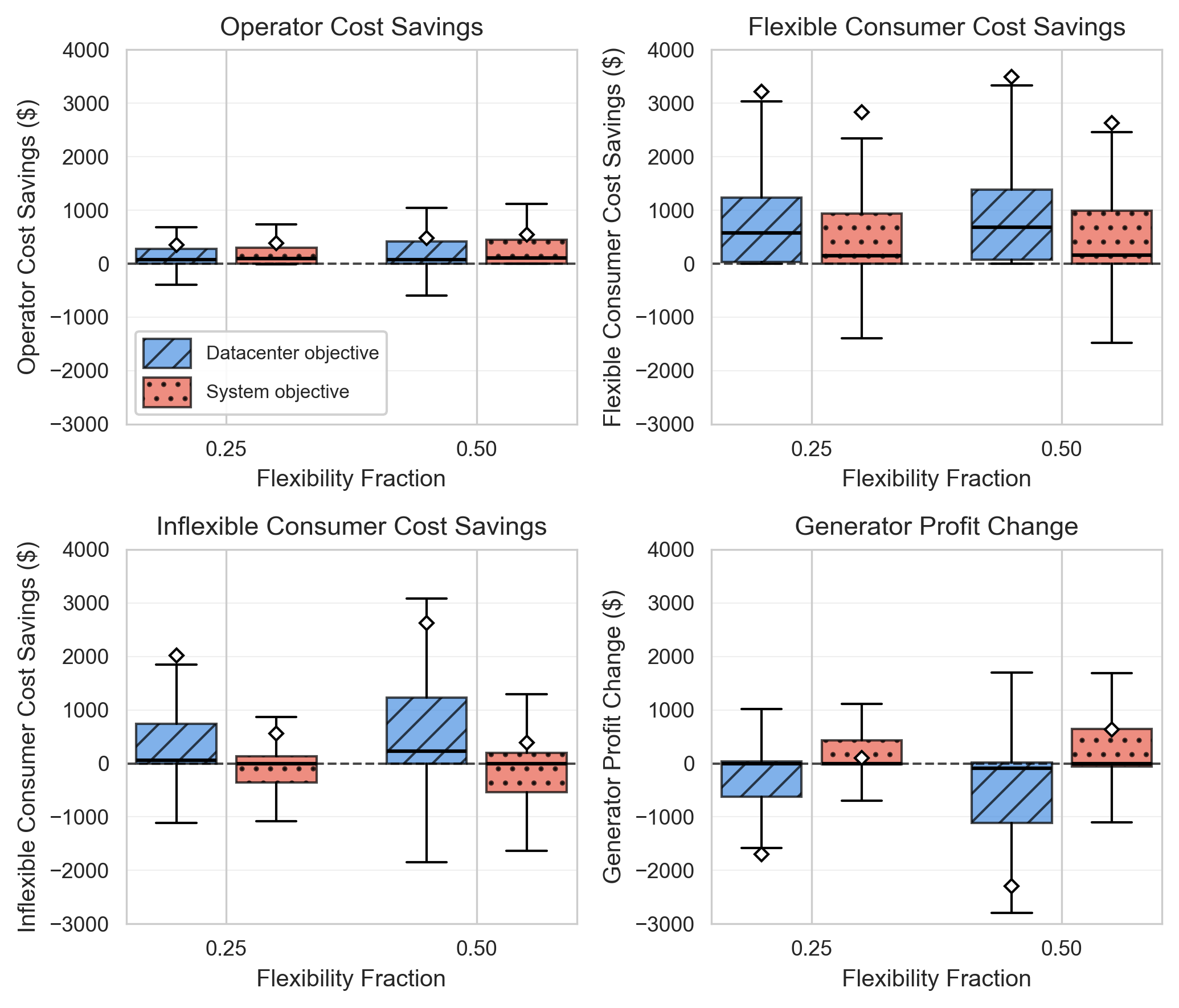}
    \caption{Hourly cost and profit impacts of spatial load shifting under $\alpha=0.25$ and $\alpha=0.5$. Boxplots compare decentralized and system-optimal outcomes relative to the no-flexibility baseline.}
    \label{fig:cost_savings}
\end{figure}

\balance

\endgroup
\end{document}